\documentclass[11pt,letterpaper]{amsart}

\usepackage[T1]{fontenc} 
\usepackage[utf8]{inputenc} 
\usepackage{amsmath,amssymb}
\usepackage{microtype} 

\usepackage{appendix}
\usepackage{abstract}

\def\dom{\mathrm {dom}\,}
\def\dr{\mathrm {d}}
\newcommand{\eps}{\epsilon}

\newcommand{\bx}{\mathbf{x}}
\newcommand{\bA}{\mathbf{A}}
\newcommand{\LL}{\mathrm{L}}
\newcommand{\CC}{\mathrm{C}}

\newcommand{\R}{\mathbb R}

\newcommand{\Scal}{\mathcal S}
\newcommand{\Hcal}{\mathcal H}

\newcommand{\norm}[1]{\lVert #1 \rVert}

\newtheorem{definition}{\bf Definition}[section]
\newtheorem{proposition}{\bf Proposition}[section]

\newtheorem{lemma}{\bf Lemma}[section]
\newtheorem{theorem}{\bf Theorem}[section]

\newtheorem{remark}{Remark}[section]

\title[Robin Aharonov-Bohm]{Remarks on Regular Approximations to the Robin Aharonov-Bohm Hamiltonian}
\author{C\'esar R. de Oliveira}
\address{C\'esar R.\ de Oliveira: Department of Mathematics, UFSCar, S\~ao Carlos, SP, 13560-970 Brazil}

\begin{document}

\maketitle

\begin{abstract}
In the space $\R^3$, for Robin parameter $L>0$, it is shown that there is a family of Schr\"odinger operators with penetrable toroidal solenoid and variable conductivity that approximates the magnetic Aharonov-Bohm operator with a Robin boundary condition at the solenoid (border). It is also shown that approximations via smooth potentials and then a barrier in the solenoid interior give Dirichlet boundary conditions. The approximations are in the strong resolvent sense and obtained through the $\Gamma$-convergence technique, and they hold for the more general setting of smooth, closed and compact surfaces and continuous and bounded magnetic potentials.
\end{abstract}

\ 


\


\

\section{Introduction}\label{sectIntrod}

There are several reasons why approximating Schrödinger operators with singular potentials or boundary conditions by operators with regular potentials is important. First, since special models often present idealizations, obtaining them as limits of `physical' Hamiltonians is conceptually significant; this justifies the use of such models and clarifies their physical interpretation. Furthermore, such approximations are often more advantageous for numerical analysis, where singular interactions can be difficult to treat directly.

The case of the magnetic Aharonov-Bohm (AB) effect, with Dirichlet boundary condition (bc) at the solenoids, was obtained as limits of penetrable solenoids and regular potentials that diverge to $+\infty$ in the solenoid interior, as discussed in \cite{Kretz1965,MVGris, deOP2008, deOP2011}. This was also intended as a way to justify the use of the Dirichlet boundary condition~\cite{deOP2008, deOP2011} in this context.

From the mathematical physics perspective, the main goal here is to discuss approximations of the magnetic AB effect with toroidal solenoids in the space, but with some Robin boundary conditions, by a family of regular operators well defined in the whole~$\R^3$. In fact, the results hold for more general closed and compact surfaces~$\mathcal S$ with nonempty interior, which could be solenoids or not, and continuous and bounded magnetic potentials.Two kinds of approximations will be considered: 
\begin{enumerate}
\item A sequence of smooth profiles that impose Robin boundary conditions on~$\Scal$, followed by a sequence of regular potentials to turn the surface impenetrable; the final result is always the presence of Dirichlet boundary conditions (see Subsection~\ref{subsecSmootApprSequen}). 
\item The second proposal also uses this diverging sequence of regular potentials in the interior~$\Scal^\circ$ of~$\Scal$ and considers a suitable sequence of variable conductivity, resulting in an impenetrable solenoid with Robin boundary condition, but only for the Robin parameter $L>0$ (see Subsection~\ref{subsectConductLaySequ}). Recall that for such surface~$\Scal$ to be considered a solenoid the space $\R^3\setminus \Scal^\circ$ cannot be simply connected.
\end{enumerate}

Next we present the more general surfaces (``solenoids'') of interest, magnetic potentials and the corresponding quadratic forms and/or operators. Let~$\Scal \subset \R^3$ be a smooth, compact, closed and oriented surface separating $\R^3$ into an interior (open and bounded) domain~$\Scal^\circ$ and an exterior (open) domain~$\Scal'$. Let $\nu(p)$ denote the unit normal vector at $p \in\Scal$ pointing towards~$\Scal'$. The surface may have many (finitely many) components; e.g., it can be a collection of disjoint tori. The (real) magnetic potential $\bA(\bx)$ is continuous and bounded (and smooth on~$\R^3\setminus\Scal$). 

In the AB context, $\Scal$ is a solenoid, $\bA$ is usually tangent to the surface~$\Scal$ and the magnetic field $\mathbf B(\bx)=\mathrm{rot}\, \bA(\bx)=0$ for $\bx\in\Scal'$, but we do not need these assumptions for the general results. Furthermore, the magnetic field is not necessarily zero in~$\Scal^\circ$, but the solenoid~$\Scal$ is impenetrable, i.e., its interior is excluded from the final Hamiltonian operators and dynamics. So the AB dynamics occur in~$\Scal'$ and it is not at all clear what is the boundary condition on the surface (solenoid)~$\Scal$. Natural possibilities are Dirichlet, Neumann and Robin bc.

Significant controversy has existed over whether potentials are real physical variables, especially concerning the presence of~$\bA$ in the Schr\"odinger operator within field-free regions created by impenetrable solenoids. This controversy was largely settled by experimental confirmations of the AB effect under conditions of essentially perfect magnetic shielding~\cite{TonomuraEtAl1982, Tonomura1986}, and more recently by  macroscopic realizations of the effect~\cite{BatelaanCaprez,BeckerBatelaan}. In this context, approximations starting from penetrable solenoids and a shielding process become relevant. Here, we aim to provide a physical justification for the use of the Robin boundary condition at the solenoid~$\Scal$.

The Robin case, with parameter $L\in\R$, is described by the quadratic form
\begin{equation}\label{eqDefQFRobin}
b^{(L)}(u) = \begin{cases} \displaystyle \int_{\Scal'} |(\nabla+i\bA)u|^2\,\dr \bx + L \int_\Scal |\gamma_D(u)|^2\,\dr \sigma, & u\in\Hcal^1(\Scal'), \\ +\infty\,, & u\in\LL^2(\Scal')\setminus \Hcal^1(\Scal'), \end{cases}
\end{equation}
where $\dr \sigma$ is the area element of~$\Scal$ (i.e., the standard surface measure given by the first fundamental form of the surface). We denote by $\gamma_D:\mathcal {\mathcal H}^1(\Scal')\to \mathcal H^{1/2}(\Scal)$ and $\gamma_N:\mathcal {\mathcal H}^2(\Scal')\to \mathcal H^{1/2}(\Scal)$, with $\gamma_N(u)=\nu\cdot \gamma_D (\nabla u)$, the Dirichlet and Neumann traces, respectively~\cite{McLean2000} (here, the Neumann trace is evaluated on the operator domains, where functions possess sufficient regularity).

The associated self-adjoint operator is the magnetic Robin Laplacian
 \begin{equation}\label{eqHLRobi}
 H^{(L)} = -(\nabla + i\bA)^2
 \end{equation}
 acting on the domain, with boundary conditions in the sense of traces,
 \[
 \dom H^{(L)} = \left\{ u \in \Hcal^2(\Scal') \mid \nu \cdot (\nabla + i\bA)u = L u \text{ on } \Scal \right\}.
 \]
 For the reader's convenience, Appendix~\ref{appActHL} shows how this operator is obtained from the form~\eqref{eqDefQFRobin}. If $\bA(\bx)$ is tangent to the surface, this bc is simplified to
\[
\gamma_N(u)=\nu\cdot \nabla u = L u=L\gamma_D(u) \text{ on } \Scal.
\]
The case $L=0$ corresponds to the Neumann bc. We will not always explicitly denote the traces by $\gamma_D$ or $\gamma_N$, consistent with the notation used above for the domain of~$H^{(L)}$.

To properly describe our approximations by regular operators and results, we anticipate some geometric preliminaries. For $\delta > 0$ sufficiently small, the tubular neighborhood
\[
\mathcal{N}_\delta = \{ \bx \in \R^3 \mid \mathrm{dist}(\bx,\Scal) < \delta \}
\]
of~$\Scal$ is diffeomorphic to the product~$\Scal \times (-\delta, \delta)$ via the map
\[
\Phi:\Scal \times (-\delta, \delta) \to \mathcal{N}_\delta, \quad \Phi(p, t) = p + t\nu(p).
\]
Here, $t$ represents the signed distance to~$\Scal $ ($t>0$ in~$\Scal'$, $t<0$ in ${\Scal^\circ}$), and this parameter characterizes the points in such neighborhoods; for example, the layer $\mathcal{S}_{\delta} = \{0 < t < \delta\}$ means
\begin{equation}\label{eqLayerSdelta}
\mathcal{S}_{\delta} = \{\bx=p+t\nu(p)\mid p\in\Scal, 0 < t < \delta\}.
\end{equation}

By a slight abuse of notation, for a function $u(\bx)$ defined in $\mathcal{N}_\delta$, we will frequently identify it with its pullback $u \circ \Phi$ and simply write $u(p, t)$ to denote the evaluation $u(p + t\nu(p))$, defining trace operators $\gamma_t(u) := u(\cdot, t)$.

The Euclidean volume element $\dr\bx $ transforms as
\[
\dr\bx = J(p, t)\,\dr \sigma(p)\,\dr t,
\]
with Jacobian $J(p, t) = \det(I + t \,\dr\nu(p))$. Since~$\Scal $ is smooth and compact, the Jacobian satisfies the uniform in~$p\in\Scal$ estimate
\begin{equation} \label{eq:jacobian}
1 - C|t| \le J(p, t) \le 1 + C|t| \quad \implies \quad J(p,t) = 1 + O(t),
\end{equation}
for $|t| < \delta$.

The gradient operator $\nabla$ in these coordinates decomposes into a normal component and a tangential component, 
\[
|\nabla u|^2 = |\partial_t u|^2 + |\nabla_\tau u|^2,
\]
where $\nabla_\tau$ is the tangential derivatives on the parallel surface $\Scal_t = \{\bx \mid \mathrm{dist}(\bx,S)=t\}$. This equality is exact for~$t=0$, and for small~$|t|$, the tangential component approximates the surface gradient on~$\Scal$ uniformly.

We introduce two notations: $\chi_{\Omega}(\bx)$ will denote the characteristic function of the set~$\Omega$, and $R_\mu(T)=(T-\mu\mathbf 1)^{-1}$ denotes the resolvent operator of the self-adjoint operator~$T$ at the scalar point~$\mu$ in its resolvent set.
\begin{remark}
This work is explicitly written for dimension three, but the mathematical arguments hold for any dimension~$N \ge 2$, provided~$\Scal$ is a  closed compact smooth hypersurface (codimension~1).
\end{remark}

\section{Approximations and results}

\subsection{Smooth approximating sequence}\label{subsecSmootApprSequen}
In this first approximating approach, we initially do not introduce a mechanism to restrict the limit function~$u$ to live solely in $\Hcal^1(\Scal')$. The Robin condition arises from a $\delta$-potential supported on~$\Scal$. In a second step we include a shielding process of the interior region.

Let $\eps_n \downarrow 0$ be a sequence of widths. Let $\varphi \in \CC_c^\infty((-1,1))$ be a profile function satisfying
\[
\varphi \ge 0, \quad \int_{-1}^1 \varphi(s)\, \dr s = 1.
\]
Define the smooth potential supported in the layer $|t| < \eps_n$ by
\[
W_{n,L}(\bx) = \begin{cases} \frac{L}{\eps_n} \varphi\left(\frac{t}{\eps_n}\right), & \bx \in \mathcal{N}_{\eps_n}, \\ 0, & \text{otherwise}. \end{cases}
\]
The sequence of associated quadratic forms $q_n^{(L)}: \LL^2(\R^3) \to [0, \infty]$ is defined by
\[
q_n^{(L)}(u) = \begin{cases} \displaystyle \int_{\R^3} |(\nabla+i\bA)u|^2\,\dr \bx + \int_{\R^3} W_{n,L} |u|^2\,\dr \bx, & u \in \Hcal^1(\R^3), \\ +\infty\,, & \text{otherwise}. \end{cases}
\] The associated self-adjoint operator is
\[
h_n^{(L)}u = -(\nabla + i\bA)^2u + W_{n,L}(\bx)u,\quad \dom h_n^{(L)}=\Hcal^2(\R^3).
\]

\begin{proposition}\label{prophLconvh}
The operator sequence $h_n^{(L)}$ converges, in the strong resolvent sense, to the operator
\[
h^{(L)} u = -(\nabla + i\mathbf{A})^2 u, \quad \text{in } \R^3 \setminus \Scal,
\]
with domain
\[
\dom h^{(L)} = \big\{ u \in \Hcal^1(\R^3) \cap \Hcal^2(\R^3 \setminus \Scal) \mid \left[ \nu \cdot(\nabla+i \mathbf{A}) u \right]_{\Scal} = L \gamma_D(u)\text{ on }\Scal \big\},
\]
where $\left[ \nu \cdot(\nabla+i \mathbf{A}) u \right]_{\Scal}$ denotes the jump of the magnetic normal derivative across the surface~$\Scal$ (a Robin-type bc).
 Equivalently,
\[
\lim_{n\to\infty}R_i(h_n^{(L)})u\to R_i(h^{(L)})u,\quad \text{for all } u\in \LL^2(\R^3).
\]
\end{proposition}
The quadratic form~$q^{(L)}$ associated with this limit operator~$h^{(L)}$ is
 \[
 q^{(L)}(u, v) = \int_{\R^3} \overline{(\nabla + i\mathbf{A})u} \cdot {(\nabla + i\mathbf{A})v}\, \dr\bx + L \int_{\Scal} \overline{u} {v}\, \dr\sigma,
 \] with $\dom q^{(L)}=\mathcal H^1(\R^3)$.
Details appear in Appendix~\ref{appActOphL}. 

A more general framework for approximating Laplace operators with Robin boundary conditions and Schr\"odinger operators with $\delta$-interactions is developed in~\cite{Holzmann2026}; see also~\cite{LotorPost} and the references therein.

Thinking of the AB effect, in particular when~$\Scal$ is a (finite) collection of tori, we try to turn the surface impenetrable by adding to $q^{(L)}$, and so to $h^{(L)}$, a sequence of diverging potentials in the surface interior. Consider then the sequence of forms defined on $\Hcal^1(\R^3)$ by
\[
\tilde{q}_m^{(L)}(u) = q^{(L)}(u) + m \int_{{\Scal^\circ}} |u|^2\,\dr \bx, 
\]
and take the limit $m \to \infty.$ The self-adjoint operator associated with $\tilde{q}_m^{(L)}$ is $\tilde h_m^{(L)}=h^{(L)}+ m\chi_{\Scal^\circ}(\bx)$. 

The magnetic operator in~$\Scal'$ with Dirichlet boundary condition is
\begin{equation}\label{eqDirichlOpera}
Du=-(\nabla + i\mathbf{A})^2 u, \quad u\in\dom D = \Hcal^2(\Scal')\cap \Hcal_0^1(\Scal').
\end{equation}
The elements of its domain are supposed to vanish in~$\Scal^\circ$. Let $P_0$ be the projection $ \LL^2(\mathbb R^3)\mapsto \LL^2(\Scal')$. Independently of the value of the Robin parameter~$L$, one has

\begin{proposition}\label{proptildehmDiric}
The sequence of operators $\tilde h_m^{(L)}$ converges in the strong resolvent sense to~$D$ in $\LL^2(\Scal')\subset \LL^2(\mathbb R^3)$, that is,
\[
\lim_{m\to\infty} R_i(\tilde h_m^{(L)})u=R_i(D)P_0u,\quad \text{for all }u\in\LL^2(\mathbb R^3).
\]
\end{proposition}

The proofs of the above two propositions are the subjects of Section~\ref{sectProsfsProps}.

\begin{remark}
The same limit $D$ in Proposition~\ref{proptildehmDiric} is obtained if instead of the operator sequence~$\tilde h_m^{(L)}$ one simultaneously includes the shielding process, that is, if one considers 
\[
h_n^{(L)}u + n\chi_{\Scal^\circ}(\bx)u = -(\nabla + i\bA)^2u + W_{n,L}(\bx)u+n\chi_{\Scal^\circ}(\bx)u,
\] as $n\to\infty$. Also here, the shielding process of the interior region~$\Scal^\circ$ results in the Dirichlet operator~$D$. This proof does not differ too much from the ones presented below and will not be discussed here.
\end{remark}

Consequently, the Robin interaction vanishes due to the continuity constraint and the hard interior obstacle, yielding the Dirichlet bc. Achieving a Robin condition with a hard obstacle ($u=0$ inside, $u \neq 0$ outside) requires $u$ to exhibit a jump discontinuity at $\mathcal{S}$, which results in infinite kinetic energy ($\int|\nabla u|^2=\infty$). To maintain regular approximations without introducing a renormalization term, we consider a variable conductivity that diminishes near the surface to suppress the gradient energy. We present this construction next.

\subsection{Conductivity layer approximating sequence}\label{subsectConductLaySequ}

This (regular) approach is physically distinct from the first one. Besides enforcing a hard obstacle $n\chi_{\Scal^\circ}(\bx)$ in the interior ${\Scal^\circ}$, we will impose in the layer $\mathcal{S}_{\eps_n} = \{0 < t < \eps_n\}$ (see~\eqref{eqLayerSdelta}) a nonconstant conductivity $ a_{n,L}(\bx)$, and suitably small in the layer and $L>0$, given by
\[
 a_{n,L}(\bx) = \begin{cases}
L\eps_n, & \bx \in \mathcal{S}_{\eps_n} , \\
1, & \bx \in \R^3\setminus \mathcal{S}_{\eps_n}.
\end{cases}
\]
The sequence of corresponding forms is 
\[
b_n^{(L)}(u) = \int_{\R^3}  a_{n,L}(\bx) |(\nabla+i\bA)u|^2\,\dr \bx + n \int_{{\Scal^\circ}} |u|^2\,\dr \bx, \quad u\in \Hcal^1(\R^3),
\]
and $b_n^{(L)}(u)=+\infty$ if $u\in\LL^2(\R^3)\setminus\Hcal^1(\R^3)$.
The idea is that the small factor $L\eps_n$ will compensate the large values of the gradient in the layer. Denote by $H_n^{(L)}$ the self-adjoint operator associated with~$b_n^{(L)}$, whose action is
\[
H_n^{(L)}u=-(\nabla+i \mathbf{A}) \cdot\big( a_{n,L}(\bx)(\nabla+i \mathbf{A}) u\big)+n \chi_{\mathcal{S}^{\circ}}(\mathbf{x}) u
\] with domain 
\[
\dom H_n^{(L)} =\big\{u\in\Hcal^1(\R^3)\mid H_n^{(L)}u\in\LL^2(\R^3) \big\}.
\]
According to Theorem~\ref{thmLimRobinMag}, in the limit $n\to\infty$, this approach gives the magnetic Robin Hamiltonian~\eqref{eqHLRobi}, with parameter $L>0$; recall that $L=0$ corresponds to Neumann bc and that $P_0$ is the projection $ \LL^2(\mathbb R^3)\mapsto \LL^2(\Scal')$.

\begin{theorem}\label{thmLimRobinMag}
Suppose that $L>0$. The sequence of operators $ H_n^{(L)}$ converges, in the strong resolvent sense, to the magnetic Robin Laplacian~$H^{(L)}$ in $\LL^2(\Scal')\subset \LL^2(\mathbb R^3)$ (see~\eqref{eqHLRobi}), that is,
\[
\lim_{n\to\infty} R_i( H_n^{(L)})u=R_i(H^{(L)})P_0u,\quad \text{for all }u\in\LL^2(\mathbb R^3).
\]
\end{theorem}

\begin{remark}
In this approach, the parameter $L$ determines the conductivity of the
thin layer through $
a_{n,L}(\bx)=L\varepsilon_n.$ 
Thus the assumption $L>0$ in Theorem~\ref{thmLimRobinMag} is essential.
Indeed, if $L<0$, then the conductivity becomes negative on the layer,
the principal part of the approximating operator is no longer elliptic,
and the associated quadratic form is not semibounded. On the other hand,
if $L=0$, the conductivity vanishes identically throughout the layer, so
the associated quadratic form is no longer closed on
$\Hcal^1(\mathbb{R}^3)$, and the First Representation Theorem does not
produce the operator $H_n^{(0)}$. In particular, the Neumann operator cannot
be obtained by simply setting $L=0$ in the approximating family
$H_n^{(L)}$.
\end{remark}

\begin{remark}
In Section~\ref{sectNeumannL0}, it is shown that for all $L_0\in\R$, $H^{(L)}$ converges, in the norm resolvent sense, to $H^{(L_0)}$ if $L\to L_0$, hence there is a kind of stability with respect to the Robin parameter.
\end{remark}

\section{$\Gamma$-convergence of quadratic forms}
\label{sectGammaTechn}

The proof of the main convergence theorem is based on the
$\Gamma$-convergence of quadratic forms. In Hilbert spaces, this notion
provides a convenient criterion for the strong resolvent convergence of
the associated self-adjoint operators. The extension of
$\Gamma$-convergence to complex Hilbert spaces was discussed
in~\cite{BdOV2014}.

Let $\Hcal$ be a complex Hilbert space, and let
$q:\dom q\subset\Hcal\to\R$ be a quadratic form. As usual, we extend $q$
to the whole space by setting
\[
q(\xi)=+\infty,
\qquad
\xi\in\Hcal\setminus\dom q.
\]
Throughout this section we consider quadratic forms that are closed and
uniformly lower bounded, namely,
\[
q(\xi)\ge\beta\|\xi\|^2,
\qquad
\xi\in\dom q,
\]
for some constant $\beta\in\R$. Such forms are lower semicontinuous
\cite[Theorem~9.3.11]{ISTQD}, and by the First Representation
Theorem \cite[Section~4.2]{ISTQD} each of them determines a unique
self-adjoint operator $T$, whose form domain is $\dom q$.

The relevant notion of convergence here is the following.

\begin{definition}
\label{defCombGammaConve}
Let $q_n$ and $q$ be closed, uniformly lower-bounded quadratic forms on
$\Hcal$. We say that $q_n$ converges to $q$ in the
\emph{combined $\Gamma$-convergence sense} if the following conditions
hold.
\begin{itemize}
\item[($\Gamma1$)] (\emph{Weak liminf inequality})
For every sequence
$\xi_n\rightharpoonup\xi$ weakly in $\Hcal$,
\[
\liminf_{n\to\infty}q_n(\xi_n)\ge q(\xi).
\]

\item[($\Gamma2$)] (\emph{Strong limsup inequality})
For every $\xi\in\Hcal$, there exists a recovery sequence
$\xi_n\to\xi$ strongly in $\Hcal$ such that
\[
\limsup_{n\to\infty}q_n(\xi_n)\le q(\xi).
\]
\end{itemize}
\end{definition}

The verification of $(\Gamma2)$ is usually the more delicate part of a
$\Gamma$-convergence proof, since it requires the explicit construction
of a recovery sequence.

Let $T_n$ and $T$ denote the self-adjoint operators associated with
$q_n$ and $q$, respectively. Since $\dom T$ need not be dense in
$\Hcal$, let
\[
\Hcal_0=\overline{\dom T}
\]
and denote by $P_0:\Hcal\to\Hcal_0$ the orthogonal projection. The following theorem is a reformulation of \cite[Theorem~1]
{BdOV2014}; for a full proof, see~\cite{dalMaso1993}.

\begin{theorem}  \label{teorGammaConvTeor}
The following assertions are equivalent.
\begin{enumerate}
\item[(i)]
$q_n$ converges to $q$ in the combined $\Gamma$-convergence sense.

\item[(ii)]
$T_n$ converges to $T$ in the strong resolvent sense,
\[
R_i(T_n)\xi
\rightarrow
R_i(T)P_0\xi,
\qquad
\forall\,\xi\in\Hcal.
\]
\end{enumerate}
\end{theorem}

In the proof of Theorem~\ref{thmLimRobinMag}, we shall therefore verify
conditions $(\Gamma1)$ and $(\Gamma2)$ in order to establish the strong
resolvent convergence of the corresponding magnetic Schr\"odinger
operators.

\section{Proof of Theorem~\ref{thmLimRobinMag}}

We prove the theorem by establishing the combined $\Gamma$-convergence of the quadratic forms
$b_n^{(L)}$ to $b^{(L)}$. The following simple inequality will be used.

\begin{lemma}\label{lemmaabdeelta}
Let $a,b\in\mathbb C$ and let $\delta>0$. Then
\[
|a+b|^2
\ge
(1-\delta)|a|^2
-\frac1\delta |b|^2.
\]
\end{lemma}

\begin{proof}
Since
\[
|a+b|^2
=
|a|^2+|b|^2+2\operatorname{Re}(\bar ab),
\]
we have
\[
|a+b|^2
\ge
|a|^2+|b|^2-2|a||b|.
\]
Young's inequality gives
\[
2|a||b|
\le
\delta |a|^2+\frac1\delta |b|^2,
\]
hence
\[
-2|a||b|
\ge
-\delta|a|^2-\frac1\delta|b|^2.
\]
Substituting this estimate yields
\[
|a+b|^2
\ge
(1-\delta)|a|^2
+\left(1-\frac1\delta\right)|b|^2
\ge
(1-\delta)|a|^2-\frac1\delta|b|^2,
\]
which proves the claim.
\end{proof}

\subsection*{Verification of $(\Gamma1)$}
Let $(u_n)\subset \LL^2(\R^3)$ satisfy $u_n\rightharpoonup u$ weakly in $\LL^2(\R^3)$ and
\[
C_0:=\sup_n b_n^{({L})}(u_n)<\infty .
\]
We must show $b^{({L})}(u)\le\liminf_n b_n^{({L})}(u_n)$. Passing to a subsequence
we may assume the liminf is a limit. Since a weakly convergent sequence is bounded,
$M:=\sup_n\|u_n\|_{\LL^2(\R^3)}<\infty$.

\smallskip
\emph{Step 1: Bound off the layer.}
Since $a_n\equiv1$ on $\R^3\setminus\mathcal S_{\eps_n}$,
\[
\int_{\R^3\setminus \mathcal S_{\eps_n}}|(\nabla+i\bA)u_n|^2\,\dr {\bx}\le C_0 .
\]
By the diamagnetic-type bound $|\nabla v|^2\le 2|(\nabla+i\bA)v|^2+2\|\bA\|_\infty^2|v|^2$,
\begin{equation}\label{eqC1bound}
\int_{\R^3\setminus \mathcal S_{\eps_n}}|\nabla u_n|^2\,\dr {\bx}
\le 2C_0+2\|\bA\|_\infty^2M^2=:C_1,
\qquad\text{uniformly in }n.
\end{equation}
Let $K\Subset\mathcal S'$ be any compact set. Since $\operatorname{dist}(K,\mathcal S)>0$, for $n$
large (depending on $K$) we have $K\subset\mathcal S'\setminus\mathcal S_{\eps_n}$, so
$\int_K|\nabla u_n|^2\le C_1$ by~\eqref{eqC1bound}. Hence $(u_n)$ is bounded in
${\mathcal H}^1(K)$ for every such $K$, and weak lower semicontinuity of the Dirichlet
energy on $K$ gives
\[
\int_K|\nabla u|^2\,\dr {\bx}\le\liminf_n\int_K|\nabla u_n|^2\,\dr {\bx}\le C_1 .
\]
Exhausting $\mathcal S'$ by an increasing sequence of compacts $K_j\uparrow\mathcal S'$ and
using monotone convergence,
\[
\int_{\mathcal S'}|\nabla u|^2\,\dr {\bx}=\lim_{j\to\infty}\int_{K_j}|\nabla u|^2\,\dr {\bx}\le C_1<\infty .
\]
Together with $u\in\LL^2(\mathcal S')$ (inherited from $u_n\rightharpoonup u$ in $\LL^2(\R^3)$),
this shows $u\in{\mathcal H}^1(\mathcal S')$, correcting the informal
``$H^1_{\rm loc}\Rightarrow H^1$'' step. The same exhaustion argument, now applied to
$(\nabla+i\bA)u_n$ (with $a_n\equiv1$ on each $K_j$ once $n$ is large), gives
\begin{equation}\label{eqouterbound}
\int_{\mathcal S'}|(\nabla+i\bA)u|^2\,\dr {\bx}
\le\liminf_n\int_{\R^3}a_n|(\nabla+i\bA)u_n|^2\,\dr {\bx} .
\end{equation}

\smallskip
\emph{Step 2: Vanishing on $\mathcal S^\circ$.}
Since $a_n\equiv1$ on $\mathcal S^\circ$ as well,
\[
\int_{\mathcal S^\circ}|(\nabla+i\bA)u_n|^2\,\dr {\bx}\le C_0,
\qquad
n\int_{\mathcal S^\circ}|u_n|^2\,\dr {\bx}\le C_0 .
\]
The second bound gives $\int_{\mathcal S^\circ}|u_n|^2\to0$, so $u_n\to0$ strongly in
$\LL^2(\mathcal S^\circ)$; since also $u_n\rightharpoonup u$, we get $u=0$ a.e.\ on
$\mathcal S^\circ$. The first bound, together with the estimate
\[
|\nabla v|^2
\le
2|(\nabla+i\bA)v|^2
+
2\|\bA\|_\infty^2|v|^2,
\]
shows that $(u_n|_{\mathcal S^\circ})$ is bounded in
${\mathcal H}^1(\mathcal S^\circ)$,   hence
$u_n|_{\mathcal S^\circ}\rightharpoonup 0$ weakly in ${\mathcal H}^1(\mathcal S^\circ)$. Since
$\mathcal S^\circ$ is bounded and $\gamma_D:{\mathcal H}^1(\mathcal S^\circ)\to\LL^2(\mathcal S)$
is compact, weak convergence in ${\mathcal H}^1(\mathcal S^\circ)$ together with the strong
$\LL^2(\mathcal S^\circ)$ convergence to $0$ forces
\begin{equation}\label{eqinnertrace}
\gamma_D\big(u_n|_{\mathcal S^\circ}\big)\rightarrow 0
\quad\text{strongly in }\LL^2(\mathcal S).
\end{equation}
Because $u_n\in{\mathcal H}^1(\R^3)$, the trace at $t=0$ is single-valued, so
$\gamma_0(u_n):=u_n(\cdot,0)$ computed from either side coincides with the trace in
\eqref{eqinnertrace}; thus $\gamma_0(u_n)\to0$ in $\LL^2(\mathcal S)$.

\smallskip
\emph{Step 3: Trace at the outer edge of the layer.}
For $n$ large ($\eps_n<\delta_0$) set $\tau_n:=\gamma_{\eps_n}(u_n)\in\LL^2(\mathcal S)$, the
trace of $u_n$ at depth $t=\eps_n$. We claim
\begin{equation}\label{eqtaun}
\tau_n\rightarrow\gamma_D(u)\quad\text{strongly in }\LL^2(\mathcal S).
\end{equation}
Fix $0<\delta_1<\delta_0$. For $n$ so large that $\eps_n<\delta_1$, the collar
$\{\eps_n<t<\delta_1\}$ lies in
$\mathcal S'\setminus\mathcal S_{\eps_n}$,
where $a_n\equiv1$. Since
\[
|\partial_tu_n|
=
|\nu\cdot\nabla u_n|
\le
|\nabla u_n|,
\]
the estimate \eqref{eqC1bound} yields
\[
\int_{\eps_n}^{\delta_1}
\|\partial_tu_n(\cdot,t)\|^2_{\LL^2(\mathcal S)}
\,dt
\le
C_1.
\]
By Cauchy-Schwarz in $t$,
\begin{equation}\label{eqterm1}
\|\tau_n-\gamma_{\delta_1}(u_n)\|_{\LL^2(\mathcal S)}
\le \Big(\delta_1\int_{\eps_n}^{\delta_1}\|\partial_tu_n(\cdot,t)\|^2_{\LL^2(\mathcal S)}\,\dr t\Big)^{1/2}
\le \sqrt{\delta_1C_1}.
\end{equation}
On the fixed compact collar $\overline{\{\delta_1/2\le t\le\delta_0\}}$, $(u_n)$ is
bounded in ${\mathcal H}^1$ by Step~1, hence $u_n\rightharpoonup u$ weakly there, and
compactness of the trace at $t=\delta_1$ gives
\begin{equation}\label{eqterm2}
\gamma_{\delta_1}(u_n)\rightarrow \gamma_{\delta_1}(u)\quad\text{strongly in }\LL^2(\mathcal S)
\quad (n\to\infty),\ \ \delta_1\text{ fixed.}
\end{equation}
Finally, since $u\in{\mathcal H}^1(\mathcal S')$, the map $t\mapsto \gamma_t(u)\in\LL^2(\mathcal S)$
is continuous on $[0,\delta_0)$, so
\begin{equation}\label{eqterm3}
\|\gamma_{\delta_1}(u)-\gamma_D(u)\|_{\LL^2(\mathcal S)}\rightarrow0
\qquad(\delta_1\to0^+).
\end{equation}
Combining \eqref{eqterm1}-\eqref{eqterm3} by the triangle inequality,
\[
\limsup_{n\to\infty}\|\tau_n-\gamma_D(u)\|_{\LL^2(\mathcal S)}
\le \sqrt{\delta_1C_1}+\|\gamma_{\delta_1}(u)-\gamma_D(u)\|_{\LL^2(\mathcal S)}
\]
for every $\delta_1\in(0,\delta_0)$; letting $\delta_1\to0^+$ proves \eqref{eqtaun}.

\smallskip
\emph{Step 4: Contribution of the thin layer.}
Write $A_\nu=\nu\cdot\bA$ and $\bA = A_\nu\nu+\bA_\tau$ for the normal and tangential decomposition of~$\bA$. 
Let
\[
g_n=(\partial_t+iA_\nu)u_n,
\]
so that, in normal coordinates,
\[
(\nabla+i\bA)u_n
=
\nu\,g_n
+
(\text{tangential component}),
\]
and therefore
\[
|(\nabla+i\bA)u_n|^2\ge |g_n|^2.
\]

Fix $p\in\mathcal S$ such that the slice
$t\mapsto u_n(\Phi(p,t))$
belongs to $\mathcal H^1(0,\eps_n)$, which holds for almost every
$p\in\mathcal S$ by the slicing theorem for Sobolev functions.
Along such a slice,
\[
\partial_tu_n=g_n-iA_\nu u_n,
\]
and hence the one-dimensional fundamental theorem of calculus yields
\[
\int_0^{\eps_n}g_n(p,t)\,\dr t
=
u_n(\Phi(p,\eps_n))
-
u_n(\Phi(p,0))
+
i\int_0^{\eps_n}
A_\nu(p,t)\,
u_n(\Phi(p,t))
\,\dr t.
\]

Applying the Cauchy-Schwarz inequality,
\[
\int_0^{\eps_n}|g_n|^2\,\dr t
\ge
\frac1{\eps_n}
\left|
\int_0^{\eps_n}g_n\,\dr t
\right|^2,
\]
and then Lemma~\ref{lemmaabdeelta} with
\[
a=
u_n(\Phi(p,\eps_n))
-
u_n(\Phi(p,0)),
\qquad
b=
i\int_0^{\eps_n}
A_\nu u_n\,\dr t,
\]
gives, for every $\delta\in(0,1)$,
\[
\int_0^{\eps_n}|g_n|^2\,\dr t
\ge
\frac{1-\delta}{\eps_n}
\left|
u_n(\Phi(p,\eps_n))
-
u_n(\Phi(p,0))
\right|^2
-
\frac{\|\bA\|_\infty^2}{\delta}
\int_0^{\eps_n}|u_n|^2\,\dr t.
\]

Integrating over $p\in\mathcal S$ and multiplying by
$a_n={L}\eps_n$, we obtain
\begin{align*}
a_n
\int_{\mathcal S}
\int_0^{\eps_n}
|g_n|^2
\,\dr t\,\dr\sigma
\ge\;&
{L}(1-\delta)
\|
\gamma_{\eps_n}(u_n)-\gamma_0(u_n)
\|_{\LL^2(\mathcal S)}^2
\\
&
-
{L}\eps_n
\frac{\|\bA\|_\infty^2}{\delta}
\int_{\mathcal S_{\eps_n}}
|u_n|^2\,\dr {\bx}.
\end{align*}

Since
\[
\int_{\mathcal S_{\eps_n}}
|u_n|^2\,\dr {\bx}
\le
M^2,
\]
choosing $
\delta=\delta_n=\sqrt{\eps_n}$ 
yields
\[
{L}\eps_n
\frac{\|\bA\|_\infty^2}{\delta_n}
\int_{\mathcal S_{\eps_n}}
|u_n|^2\,\dr {\bx}
=
O(\sqrt{\eps_n})
\rightarrow0,
\]
whereas $
1-\delta_n\rightarrow1.$

Finally, using the Jacobian estimate~\eqref{eq:jacobian} to pass from tubular coordinates to the Euclidean measure,
\[
\liminf_{n\to\infty}
a_n
\int_{\mathcal S_{\eps_n}}
|(\nabla+i\bA)u_n|^2\,\dr {\bx}
\ge
{L}
\lim_{n\to\infty}
\|
\gamma_{\eps_n}(u_n)-\gamma_0(u_n)
\|_{\LL^2(\mathcal S)}^2.
\]

By Step~2, $
\gamma_0(u_n) \rightarrow0,$  and by Step~3, $
\gamma_{\eps_n}(u_n)=\tau_n
\rightarrow\gamma_D(u)$  strongly in $\LL^2(\mathcal S)$.
Therefore,
\[
\gamma_{\eps_n}(u_n)-\gamma_0(u_n)
\rightarrow
\gamma_D(u)
\]
strongly in $\LL^2(\mathcal S)$, and hence
\begin{equation}
\label{eqlayerfinal}
\liminf_{n\to\infty}
a_n
\int_{\mathcal S_{\eps_n}}
|(\nabla+i\bA)u_n|^2\,\dr {\bx}
\ge
{L}
\int_{\mathcal S}
|\gamma_D(u)|^2\,\dr\sigma.
\end{equation}

\smallskip
Finally, decomposing $\R^3=\mathcal S^\circ\cup\mathcal S\cup
\mathcal S_{\eps_n}\cup(\mathcal S'\setminus\mathcal S_{\eps_n})$ and using
\eqref{eqouterbound} and \eqref{eqlayerfinal} (the barrier term over $\mathcal S^\circ$ is
nonnegative and can be dropped),
\[
\liminf_{n\to\infty}b_n^{({L})}(u_n)
\ \ge\
\int_{\mathcal S'}|(\nabla+i\bA)u|^2\,\dr {\bx}
+{L}\int_{\mathcal S}|\gamma_D(u)|^2\,\dr\sigma
= b^{({L})}(u).
\]
This proves $(\Gamma1)$.

\subsection*{Verification of $(\Gamma2)$.}
Let $u\in{\mathcal C}=\CC^\infty(\overline{\mathcal S'})\cap{\mathcal H}^1(\mathcal S')$, a dense
subset of ${\mathcal H}^1(\mathcal S')$. Extend $u$ by zero to $\mathcal S^\circ$ and define
the recovery sequence
\[
u_n(x)=
\begin{cases}
0, & x\in\mathcal S^\circ,\\[1mm]
\dfrac{t}{\eps_n}\,\gamma_{\eps_n}(u)(p), & x=\Phi(p,t)\in\mathcal S_{\eps_n},\\[2mm]
u(x), & x\in\mathcal S'\setminus\mathcal S_{\eps_n}.
\end{cases}
\]
By construction $u_n\in{\mathcal H}^1(\R^3)$ and $u_n\to u$ strongly in $\LL^2(\R^3)$: outside
the layer $u_n=u$ exactly, on $\mathcal S^\circ$ both are $0$, and on the layer
$\|u_n\|_{\LL^2(\mathcal S_{\eps_n})}\to0$ because $u$ is bounded and
$|\mathcal S_{\eps_n}|\to0$. (This is all that Definition~\ref{defCombGammaConve} requires;
we do {not} claim, and it is in fact false, that $u_n\to u$ strongly in
${\mathcal H}^1(\R^3)$; see the computation below, where $\nabla u_n$ blows up like
$1/\eps_n$ inside the shrinking layer while $\nabla u=O(1)$ there. It is exactly this
blow-up, tempered by the vanishing conductivity $a_n$, that produces the nontrivial Robin
term in the limit.)

Inside the layer,
\[
\nabla u_n=\nu\,\frac{\gamma_{\eps_n}(u)(p)}{\eps_n}+\frac{t}{\eps_n}\nabla_\tau \gamma_{\eps_n}(u)(p),
\]
so that
\[
(\nabla+i\bA)u_n(p,t)=\nu\,\frac{\gamma_{\eps_n}(u)(p)}{\eps_n}+R_n(p,t),
\]
\[
R_n(p,t):=\frac{t}{\eps_n}\nabla_\tau \gamma_{\eps_n}(u)(p)
+i\bA\big(\Phi(p,t)\big)\,\frac{t}{\eps_n}\gamma_{\eps_n}(u)(p).
\]
Since $u\in \CC^\infty(\overline{\mathcal S'})$, both $\nabla_\tau u$ and $u$ are bounded on
the compact set $\mathcal S\times[0,\delta_0]$; since also $\bA$ is bounded on $\R^3$ and
$t/\eps_n\in(0,1)$ on the layer, $R_n$ satisfies
\[
\sup_{n}\ \sup_{(p,t)\in\mathcal S\times(0,\eps_n)}|R_n(p,t)|<\infty .
\]
Consequently, expanding the square and using $|\nu\,a|^2=|a|^2$,
\begin{align*}
|(\nabla+i\bA)u_n(p,t)|^2
&=
\frac{|\gamma_{\eps_n}(u)(p)|^2}{\eps_n^2}
+
2\operatorname{Re}\!\Big(
\overline{\nu\,\tfrac{\gamma_{\eps_n}(u)(p)}{\eps_n}}
\cdot R_n(p,t)
\Big)
+
|R_n(p,t)|^2
\nonumber\\
&=
\frac{|\gamma_{\eps_n}(u)(p)|^2}{\eps_n^2}
+
O\!\left(\frac{1}{\eps_n}\right),
\label{eqLayerExpansion}
\end{align*}
uniformly in $p\in\mathcal S$, $t\in(0,\eps_n)$, since the cross term is
$O(|\gamma_{\eps_n}(u)(p)|/\eps_n)=O(1/\eps_n)$ and $|R_n|^2=O(1)$.

Since $a_n={L}\eps_n$, integrating over $\mathcal S_{\eps_n}$ (using $J=1+O(t)$) gives
\begin{align*}
\int_{\mathcal S_{\eps_n}}a_n|(\nabla+i\bA)u_n|^2\,\dr {\bx}
&={L}\eps_n\int_{\mathcal S}\int_0^{\eps_n}
\Big(\frac{|\gamma_{\eps_n}(u)(p)|^2}{\eps_n^2}+O\big(\tfrac1{\eps_n}\big)\Big)\,(1+O(t))\,\dr t\,\dr\sigma
\\ &={L}\int_{\mathcal S}|\gamma_{\eps_n}(u)(p)|^2\,\dr\sigma+o(1).
\end{align*}
As $\eps_n\to0$, $\gamma_{\eps_n}(u)(p)\to\gamma_D(u)(p)$ uniformly on $\mathcal S$ (since $u\in
C^\infty(\overline{\mathcal S'})$), so
\[
\lim_{n\to\infty}\int_{\mathcal S_{\eps_n}}a_n|(\nabla+i\bA)u_n|^2\,\dr {\bx}
={L}\int_{\mathcal S}|\gamma_D(u)|^2\,\dr\sigma .
\]

Outside the layer $u_n=u$, so by dominated convergence
\[
\int_{\mathcal S'\setminus\mathcal S_{\eps_n}}|(\nabla+i\bA)u_n|^2\,\dr {\bx}
\rightarrow\int_{\mathcal S'}|(\nabla+i\bA)u|^2\,\dr {\bx} .
\]
Since $u_n\equiv0$ on $\mathcal S^\circ$, the barrier term vanishes identically for every
$n$. Adding the three contributions,
\[
\limsup_{n\to\infty}b_n^{({L})}(u_n)
=\int_{\mathcal S'}|(\nabla+i\bA)u|^2\,\dr {\bx}+{L}\int_{\mathcal S}|\gamma_D(u)|^2\,\dr\sigma
= b^{({L})}(u),
\]
which is even an equality, and in particular $(\Gamma2)$ holds for $u\in{\mathcal C}$. 

Since $\mathcal C=\CC^\infty(\overline{\mathcal S'})\cap{\mathcal H}^1(\mathcal S')$ is dense
in ${\mathcal H}^1(\mathcal S')$, fix $u\in{\mathcal H}^1(\mathcal S')$ (extended by $0$ to
$\mathcal S^\circ$, as an element of $\LL^2(\R^3)$) and a sequence $u_k\in\mathcal C$ with
$u_k\to u$ strongly in ${\mathcal H}^1(\mathcal S')$. Since $\bA$ is bounded, the map
$v\mapsto(\nabla+i\bA)v$ is a bounded operator ${\mathcal H}^1(\mathcal S')\to\LL^2(\mathcal
S')$, and $\gamma_D:{\mathcal H}^1(\mathcal S')\to\LL^2(\mathcal S)$ is bounded (trace
theorem); hence the quadratic form $b^{({L})}$ is continuous with respect to strong
${\mathcal H}^1(\mathcal S')$-convergence, so
\[
b^{({L})}(u_k)\rightarrow b^{({L})}(u)\qquad(k\to\infty).
\]

For each fixed $k$, the construction above (applied to $u_k\in\mathcal C$) produces a
sequence $v_n^k\to u_k$ strongly in $\LL^2(\R^3)$ as $n\to\infty$, with
\[
\limsup_{n\to\infty}b_n^{({L})}(v_n^k)\le b^{({L})}(u_k).
\]
Fix $k$. By definition of $\limsup$ and of strong convergence, there exists $N_k$ such that
\[
n\ge N_k\ \Longrightarrow\
\|v_n^k-u_k\|_{\LL^2(\R^3)}<\frac1k,
\qquad
b_n^{({L})}(v_n^k)\le b^{({L})}(u_k)+\frac1k .
\]
Choose $n_1<n_2<\cdots$ inductively with $n_k\ge N_k$ and $n_k\to\infty$. Define
$k(n):=\max\{k\ge1\mid n_k\le n\}$ for $n\ge n_1$ (so $k(n)\to\infty$ as $n\to\infty$), and set
the diagonal sequence
\[
w_n:=v_n^{\,k(n)},\qquad n\ge n_1 .
\]
By construction $n\ge n_{k(n)}\ge N_{k(n)}$, so
\[
\|w_n-u_{k(n)}\|_{\LL^2(\R^3)}<\frac1{k(n)},
\qquad
b_n^{({L})}(w_n)\le b^{({L})}\!\big(u_{k(n)}\big)+\frac1{k(n)} .
\]
Since $k(n)\to\infty$, the right sides of these two estimates give, respectively,
\[
\|w_n-u\|_{\LL^2(\R^3)}
\le \|w_n-u_{k(n)}\|_{\LL^2(\R^3)}+\|u_{k(n)}-u\|_{\LL^2(\R^3)}
\rightarrow0,
\]
using $u_k\to u$ in ${\mathcal H}^1(\mathcal S')\hookrightarrow\LL^2(\R^3)$, and
\[
\limsup_{n\to\infty}b_n^{({L})}(w_n)
\le\limsup_{n\to\infty}\Big[b^{({L})}\!\big(u_{k(n)}\big)+\frac1{k(n)}\Big]
= b^{({L})}(u),
\]
using $b^{({L})}(u_k)\to b^{({L})}(u)$. Thus $w_n$ is a recovery sequence for $u$,
and $(\Gamma2)$ holds for every $u\in{\mathcal H}^1(\mathcal S')$.

Finally, for $\xi\in\LL^2(\R^3)$ with $\xi\notin{\mathcal H}^1(\mathcal S')$ under the above
identification  (either $\xi\not\equiv0$ on $\mathcal S^\circ$, or $\xi|_{\mathcal S'}\in
\LL^2(\mathcal S')\setminus{\mathcal H}^1(\mathcal S')$) one has $b^{({L})}(\xi)=+\infty$
by definition of $b^{({L})}$, so the constant sequence $\xi_n\equiv\xi$ trivially
satisfies $\xi_n\to\xi$ strongly in $\LL^2(\R^3)$ and $\limsup_n b_n^{({L})}(\xi_n)\le
+\infty=b^{({L})}(\xi)$. Hence $(\Gamma2)$ holds on all of $\LL^2(\R^3)$.

Since both $(\Gamma1)$ and $(\Gamma2)$ hold, $b_n^{({L})}$ converges to $b^{({L})}$
in the combined $\Gamma$-convergence sense. By Theorem~\ref{teorGammaConvTeor}, the
associated self-adjoint operators satisfy
\[
R_i\big(H_n^{({L})}\big)u\rightarrow R_i\big(H^{({L})}\big)P_0u,
\qquad u\in\LL^2(\R^3),
\]
which is the strong resolvent convergence asserted in Theorem~\ref{thmLimRobinMag}.

\section{Proofs of Propositions~\ref{prophLconvh} and~\ref{proptildehmDiric}} \label{sectProsfsProps}

\subsection{Proof of Proposition~\ref{prophLconvh}}

We establish the combined $\Gamma$-convergence of the quadratic forms $q_n^{(L)}$ to $q^{(L)}$.

\subsection*{Verification of $(\Gamma1)$}

Let $u_n \rightharpoonup u$ weakly in $\LL^2(\R^3)$ and suppose that $\liminf q_n^{(L)}(u_n) < \infty$. We need to check that $u \in \Hcal^1(\R^3)$ and
\[
\liminf_{n \to \infty} q_n^{(L)}(u_n) \ge q^{(L)}(u).
\]

{\em Regularity of the limit.}
Since $W_{n,L} \ge 0$, the magnetic kinetic energy $\int_{\R^3} |(\nabla+i\bA)u_n|^2\,\dr \bx$ is uniformly bounded. Because the magnetic potential $\bA$ is bounded and $(u_n)$ is bounded in $\LL^2(\R^3)$, it follows that the standard gradient is also uniformly bounded, i.e., $\int_{\R^3} |\nabla u_n|^2\,\dr \bx \le C$. This implies that $u_n \rightharpoonup u$ weakly in $\Hcal^1(\R^3)$. By the weak lower semicontinuity of the norm in $\Hcal^1$,
\begin{equation}\label{eqKinetEnerEstim}
\int_{\R^3} |(\nabla+i\bA)u|^2\,\dr \bx \le \liminf_{n \to \infty} \int_{\R^3} |(\nabla+i\bA)u_n|^2\,\dr \bx.
\end{equation}

{\em Analysis of the $W_{n,L}$ term.}
We must analyze the term $I_n = \int_{\R^3} W_{n,L} |u_n|^2\,\dr \bx$. Using tubular coordinates $(p, t)$, one has
\[
I_n = \int_\Scal \int_{-\eps_n}^{\eps_n} \frac{L}{\eps_n} \varphi\left(\frac{t}{\eps_n}\right) |\gamma_t(u_n)(p)|^2 J(p, t)\,\dr t\,\dr \sigma(p).
\]
We perform the change of variables $t = \eps_n s$ so that $s \in (-1, 1)$ and
\[
I_n = \int_{-1}^{1} L \varphi(s) \left( \int_\Scal |\gamma_{\eps_ns}(u_n)(p)|^2 J(p, \eps_n s)\, \dr \sigma(p) \right) \dr s.
\]

{\em Compactness of traces.}
Recall that $\mathcal{N}_\delta$ denotes a fixed tubular neighborhood of $\Scal$ where the normal coordinates are valid. Since the sequence $(u_n)$ is bounded globally in $\Hcal^1(\R^3)$, its restriction is bounded in $\Hcal^1(\mathcal{N}_\delta)$. The compactness of the trace operator $\gamma_D: \Hcal^1(\mathcal{N}_\delta) \to \LL^2(\Scal)$ implies that, up to a subsequence, $\gamma_0(u_n) \to \gamma_D(u)$ strongly in $\LL^2(\Scal)$. To rigorously evaluate the trace on the moving parallel surfaces $t = \eps_n s$, we use the Fundamental Theorem of Calculus along the normal direction, i.e.,
\[
|\gamma_{\eps_ns}(u_n)(p) - \gamma_0(u_n)(p)|^2 \le \left( \int_0^{\eps_n s} |\partial_t u_n(p, \tau)|\,\dr \tau \right)^2 \le \eps_n \int_{-\eps_n}^{\eps_n} |\partial_t u_n(p, \tau)|^2\,\dr \tau.
\]
Integrating over $\Scal$, we obtain
\[
\norm{\gamma_{\eps_ns}(u_n) - \gamma_0(u_n)}_{\LL^2(\Scal)}^2 \le \eps_n \norm{\nabla u_n}_{\LL^2(\mathcal{N}_{\eps_n})}^2 \le C \eps_n \xrightarrow[n\to\infty]{} 0,
\]
uniformly for $s \in (-1, 1)$. Therefore, $\gamma_{\eps_ns}(u_n)$ converges strongly to $\gamma_D(u)$ in $\LL^2(\Scal)$.

{\em Convergence.}
Since $J(p, \eps_n s) \to 1$ uniformly, we have
\[
\lim_{n \to \infty} \int_\Scal |\gamma_{\eps_ns}(u_n)(p)|^2 J(p, \eps_n s)\,\dr \sigma(p) = \int_\Scal |\gamma_D(u)(p)|^2\,\dr \sigma(p).
\]
Because $(u_n)$ has uniformly bounded traces in the tubular neighborhood, the inner integral is uniformly bounded. By dominated convergence,
\begin{align*}
\lim_{n \to \infty} I_n &= L \int_{-1}^1 \varphi(s) \left( \lim_{n \to \infty} \int_\Scal |\gamma_{\eps_ns}(u_n)(p)|^2 J(p, \eps_n s)\,\dr \sigma(p) \right) \dr s \\
&= L \left( \int_{-1}^1 \varphi(s) \dr s \right) \int_\Scal |\gamma_D(u)|^2\,\dr \sigma.
\end{align*}
Since $\int_{-1}^1 \varphi(s) \dr s = 1$, we obtain exactly $L \int_\Scal |\gamma_D(u)|^2\,\dr \sigma$. Combining this with the kinetic energy estimate~\eqref{eqKinetEnerEstim} concludes the proof of the $(\Gamma1)$ condition.

\subsection*{Verification of $(\Gamma2)$}
For any $u\in\Hcal^1(\R^3)$, we must construct a recovery sequence $(u_n)$ such that $u_n\to u$ and $\limsup_{n\to\infty} q_n^{(L)}(u_n)\le q^{(L)}(u)$.

We choose the constant recovery sequence
\[
u_n=u.
\]
Since the magnetic kinetic term does not depend on $n$, it remains unchanged. Thus, it suffices to analyze the concentrating potential term
\[
I_n:=\int_{\R^3}W_{n,L}(\bx)|u(\bx)|^2\,\dr\bx.
\]

For $n$ sufficiently large, one has $\eps_n<\delta$, ensuring that the support of $W_{n,L}$ is contained entirely within the tubular neighborhood $\mathcal N_{\eps_n}$. Using the normal coordinates $(p,t)$ and the volume element $\dr\bx=J(p,t)\,\dr\sigma(p)\,\dr t$ introduced in the geometric preliminaries, we recall that $J(p,t)=1+O(t)$ uniformly for $|t|<\delta$.
Using these coordinates and the definition of $W_{n,L}$, we obtain
\[
I_n=
\int_\Scal
\int_{-\eps_n}^{\eps_n}
\frac{L}{\eps_n}
\varphi\!\left(\frac{t}{\eps_n}\right)
|\gamma_t(u)(p)|^2
J(p,t)\,\dr t\,\dr\sigma(p).
\]
Changing variables $t=\eps_n s$, with $s\in(-1,1)$, yields
\[
I_n=
L\int_{-1}^{1}\varphi(s)
\left(
\int_\Scal
|\gamma_{\eps_ns}(u)(p)|^2
J(p,\eps_n s)\,\dr\sigma(p)
\right)\dr s.
\]

We now analyze the inner integral. For $|t|<\delta$, the trace operators on the parallel surfaces $\Scal_t:=\{\,\Phi(p,t)\mid p\in\Scal\,\}$ are uniformly bounded from $\Hcal^1(\R^3)$ into $\LL^2(\Scal)$. Moreover, in tubular coordinates, the traces depend continuously on $t$. Hence, for every fixed $s\in(-1,1)$,
\[
\gamma_{\eps_ns}(u)\to\gamma_D(u)
\qquad
\text{in }\LL^2(\Scal),
\]
as $n\to\infty$. Since also
\[
J(p,\eps_n s)\to1
\]
uniformly in $(p,s)\in\Scal\times[-1,1]$, it follows that
\[
\lim_{n\to\infty}
\int_\Scal
|\gamma_{\eps_ns}(u)(p)|^2
J(p,\eps_n s)\,\dr\sigma(p)
=
\int_\Scal|\gamma_D(u)(p)|^2\,\dr\sigma(p).
\]

Furthermore, the family of trace operators on $\Scal_t$, $|t|<\delta$, is uniformly bounded, so the inner integral is bounded uniformly in $n$ and $s$. Since $\varphi\in \LL^1(-1,1)$, the Dominated Convergence Theorem applies and gives
\begin{align*}
\lim_{n\to\infty} I_n
&=
L\int_{-1}^{1}\varphi(s)
\left(
\int_\Scal|\gamma_D(u)|^2\,\dr\sigma
\right)\dr s \\
&=
L\left(\int_{-1}^{1}\varphi(s)\,\dr s\right)
\int_\Scal|\gamma_D(u)|^2\,\dr\sigma.
\end{align*}
Using the normalization $
\int_{-1}^{1}\varphi(s)\,\dr s=1,
$ 
we conclude that
\[
\lim_{n\to\infty} I_n
=
L\int_\Scal|\gamma_D(u)|^2\,\dr\sigma.
\]

Therefore,
\[
\limsup_{n\to\infty} q_n^{(L)}(u_n)
=
\lim_{n\to\infty} q_n^{(L)}(u)
=
q^{(L)}(u),
\]
which proves the limsup inequality.

\subsection{Proof of Proposition~\ref{proptildehmDiric}}

We establish the combined $\Gamma$-convergence of the quadratic forms $\tilde{q}_m^{(L)}$ to the Dirichlet form $q_D$ in the exterior, which is given by
\[
q_D(u) = \begin{cases} 
\displaystyle \int_{\Scal'} |(\nabla+i\bA)u|^2\,\dr \bx, & u \in \Hcal_0^1(\Scal') , \\
+\infty, & \text{otherwise},
\end{cases}
\] 
and recall that the functions~$u$ are regarded as elements of~$\LL^2(\R^3)$, extended by zero to~$\Scal^\circ$.

\subsection* {Verification of ($\Gamma1$).}
Let $u_m \rightharpoonup u$ weakly in $\LL^2(\R^3)$. We must show that $\liminf_{m \to \infty} \tilde{q}_m^{(L)}(u_m) \ge q_D(u)$. Without loss, suppose that 
\[
\liminf_{m \to \infty} \tilde{q}_m^{(L)}(u_m) < \infty.
\] 
Passing to a subsequence, we assume that the limit exists and is bounded. As in the proof of Theorem~\ref{thmLimRobinMag}, one gets $u = 0$ almost everywhere in~$\Scal^\circ$. 

Since $(u_m)$ is weakly convergent in $\LL^2(\R^3)$, it is uniformly bounded in the $\LL^2$-norm. Combined with the boundedness of the magnetic kinetic energy and the fact that the magnetic potential $\bA$ is bounded, it follows that the standard gradient is also bounded. Thus, $(u_m)$ is bounded in $\Hcal^1(\R^3)$ and converges weakly to $u$ in $\Hcal^1(\R^3)$ (and thus $u \in \Hcal^1(\R^3)$). By the continuity of the trace operator $\gamma_D: \Hcal^1(\R^3) \to \Hcal^{1/2}(\Scal)$ and the compact embedding of the trace into $\LL^2(\Scal)$, the traces converge strongly
\[
\gamma_D(u_m) \to \gamma_D(u) \quad \text{in } \LL^2(\Scal).
\]

Since $u=0$ a.e.\ in~$\Scal^\circ$ and $u\in\Hcal^1(\R^3)$ has a single-valued trace on~$\Scal$ (the traces obtained from~$\Scal^\circ$ and from~$\Scal'$ must coincide), it follows that $\gamma_D(u)=0$; hence $u|_{\Scal'}\in\Hcal_0^1(\Scal')$, so $u\in\dom q_D$. Consequently
\[
L\int_{\Scal}|\gamma_D(u_m)|^2\,\dr\sigma\rightarrow L\int_{\Scal}|\gamma_D(u)|^2\,\dr\sigma=0
\qquad(m\to\infty),
\]
regardless of the sign of~$L$. Since $(u_m)$ is bounded in $\Hcal^1(\R^3)$ and $u_m\rightharpoonup u$ weakly, the magnetic kinetic energy is weakly lower semicontinuous,
\[
\int_{\R^3}|(\nabla+i\bA)u|^2\,\dr\bx\le\liminf_{m\to\infty}\int_{\R^3}|(\nabla+i\bA)u_m|^2\,\dr\bx,
\]
and since $u=0$ a.e.\ in~$\Scal^\circ$, the left-hand side equals $\int_{\Scal'}|(\nabla+i\bA)u|^2\,\dr\bx=q_D(u)$. Dropping the nonnegative barrier term $m\int_{\Scal^\circ}|u_m|^2\,\dr\bx\ge0$, and combining the last two displayed limits (the boundary term having an actual limit, not merely a liminf),
\[
\liminf_{m\to\infty}\tilde{q}_m^{(L)}(u_m)
\ge
\liminf_{m\to\infty}\int_{\R^3}|(\nabla+i\bA)u_m|^2\,\dr\bx
+
\lim_{m\to\infty} L\int_{\Scal}|\gamma_D(u_m)|^2\,\dr\sigma
\ge
q_D(u),
\]
which proves $(\Gamma1)$.

\subsection* {Verification of ($\Gamma2$).}
Let $u \in \dom q_D$. This means $u \in \Hcal_0^1(\Scal')$ and $u=0$ in~$\Scal^\circ$. We must find a recovery sequence $u_m \to u$ such that $\limsup \tilde{q}_m^{(L)}(u_m) \le q_D(u)$.
We choose the constant sequence $u_m = u$ for all $m$.
\begin{itemize}
 \item Since $u \in \Hcal_0^1(\Scal')$, we can naturally extend it by zero to the interior region, so $u \in \Hcal^1(\R^3)$. Thus $u_m \in \dom \tilde{q}_m^{(L)}$.
 \item The interior integral vanishes, so the term $m \int_{\Scal^\circ} |u|^2 \dr \bx = 0$.
 \item The trace vanishes: $\gamma_D(u) = 0$ on~$\Scal$, so the Robin term $L \int_{\Scal} |\gamma_D(u)|^2 \dr \sigma = 0$.
\end{itemize}
Substituting into the form, it follows that
\[
\tilde{q}_m^{(L)}(u) = \int_{\Scal'} |(\nabla+i\bA)u|^2 \dr \bx + 0 + 0 = q_D(u).
\]
Thus, $\lim_{m \to \infty} \tilde{q}_m^{(L)}(u_m) = q_D(u)$, confirming the $(\Gamma2)$ condition. The proof of the proposition is complete.

\section{Stability of $H^{(L)}$}\label{sectNeumannL0}
We apply an abstract perturbation result (Theorem~\ref{thmAbstractNorm}) that appears in~\cite{Verri2012} to prove the stability of the operator $H^{(L)}$ with respect to the parameter $L$ (Proposition~\ref{propStabiliHL}).

\begin{theorem}[Theorem 1 in \cite{Verri2012}] \label{thmAbstractNorm}
 Let $a_\epsilon$ and $q_\epsilon$ be two sequences of positive closed sesquilinear forms in a Hilbert space $\mathcal{H}$ with common domain $\mathcal{D}$. Let $A_\epsilon$ and $Q_\epsilon$ be the associated self-adjoint operators. Suppose that there exists $\lambda > 0$ such that $a_\epsilon, q_\epsilon \ge \lambda$ and
 \begin{equation} \label{eq:condition}
 |a_\epsilon(\psi) - q_\epsilon(\psi)| \le \delta(\epsilon) q_\epsilon(\psi), \quad \forall \psi \in \mathcal{D},
 \end{equation}
 where $\delta(\epsilon) \to 0$ as $\epsilon \to 0$. Then, for $\epsilon$ small enough, there exists $K > 0$ such that, for the resolvent operators,
 \[
 \norm{A_\epsilon^{-1} - Q_\epsilon^{-1}} \le K \delta(\epsilon).
 \]
\end{theorem}

\begin{proposition}\label{propStabiliHL}
 Let $L,L_0 \ge 0$. The operator $H^{(L)}$ converges to $H^{(L_0)}$ in the norm resolvent sense as $L \to L_0$. Specifically, for $\mu > 0$, there is~$C>0$ so that
 \[
 \norm{(H^{(L)} + \mu)^{-1} - (H^{(L_0)} + \mu)^{-1}} \le C |L - L_0|.
 \]
\end{proposition}

\begin{proof}
 We verify the conditions of Theorem~\ref{thmAbstractNorm} (the precise convergence to zero in $\eps\to0$ is not relevant). The Hilbert space is $\mathcal{H} = \LL^2(\Scal')$ and the common form domain is $\mathcal{D} = \Hcal^1(\Scal')$. Note that for $\mu > 0$ the shifted quadratic form 
 \[
 g^{(L)}(u) := b^{(L)}(u) + \mu \norm{u}^2 = \int_{\Scal'} |(\nabla+i\mathbf{A})u|^2 \,\dr\bx + L \int_{\Scal} |u|^2 \,\dr\sigma + \mu \norm{u}^2
 \]
 is strictly positive $g^{(L)}\ge\mu$, and the same holds for $g^{(L_0)}(u)$. Consider the difference between the forms for any $u \in \Hcal^1(\Scal')$,
\[
|g^{(L)}(u) - g^{(L_0)}(u)| = |L - L_0| \int_{\Scal} |u|^2 \dr \sigma = |L - L_0|\, \norm{u}_{\LL^2(\Scal)}^2.
\]
To apply Theorem~\ref{thmAbstractNorm}, we must bound this term by the quadratic form $g^{(L_0)}(u)$. By the Trace Theorem, the boundary $\LL^2$-norm is controlled by the $\Hcal^1$-Sobolev norm. Additionally, because the magnetic potential $\bA$ is bounded, the shifted form $g^{(L_0)}$ is coercive with respect to the standard $\Hcal^1(\Scal')$ norm. Indeed, dropping the non-negative boundary term (since $L_0 \ge 0$), one has $g^{(L_0)}(u) \ge \norm{(\nabla+i\bA)u}^2 + \mu \norm{u}^2$, which is equivalent to the standard $\Hcal^1$ norm. Thus, there exists a constant $K > 0$ such that
\[
\norm{u}_{\LL^2(\Scal)}^2 \le C_{\mathrm{tr}} \norm{u}_{\Hcal^1(\Scal')}^2 \le K \, g^{(L_0)}(u).
\]
Substituting this yields the required estimate
\[
|g^{(L)}(u) - g^{(L_0)}(u)| \le |L - L_0| K \, g^{(L_0)}(u).
\] 
This satisfies condition \eqref{eq:condition} with $\delta(L) = K |L - L_0|$. Applying Theorem~\ref{thmAbstractNorm}, we conclude the proof.
\end{proof}

\begin{remark}
By considering $\mu>0$ large enough, it is possible to generalize Proposition~\ref{propStabiliHL} to cover all values of $L,L_0\in\mathbb R$. Indeed, if $L$ or $L_0$ is negative, the boundary term in the quadratic form $b^{(L)}$ is negative, and the form is no longer positive definite. However, as demonstrated in Appendix~\ref{appActHL}, the form $b^{(L)}$ is uniformly bounded from below by some constant $\beta < 0$. By choosing a shift $\mu > |\beta|$ sufficiently large (depending on the maximum absolute values of $L$ and $L_0$), the shifted forms $g^{(L)}(u) = b^{(L)}(u) + \mu \norm{u}^2$ and $g^{(L_0)}(u)$ become strictly positive, satisfying $g^{(L)}, g^{(L_0)} \ge \lambda > 0$. Since this shift restores the coercivity of $g^{(L_0)}$ with respect to the $\Hcal^1(\Scal')$ norm, the trace inequality bound $|g^{(L)}(u) - g^{(L_0)}(u)| \le |L - L_0| K \, g^{(L_0)}(u)$ holds just as in the positive case. Theorem~\ref{thmAbstractNorm} can then be applied to the shifted operators, yielding the norm resolvent convergence.
\end{remark}

\appendix

\section{The operator $H^{(L)}$}\label{appActHL}
 Let's see how the magnetic Robin operator $H^{(L)}$ follows from the quadratic form~\eqref{eqDefQFRobin}. Keep the notation $b^{(L)}$ for the corresponding sesquilinear form $b^{(L)}(u,v)$. But first we show that, for any $L \in \R$, the form $b^{(L)}$ is closed and lower bounded. We present the proof in a more general setting for which $\bA$ is supposed to be only locally bounded.

\subsection*{Lower Boundedness}
We need to show that there exists a constant $\beta \in \R$ such that $b^{(L)}(u) \geq \beta \norm{u}_{\LL^2(\Scal')}^2$, for all $u \in \Hcal^1(\Scal')$.

If $L \geq 0$, the boundary term is nonnegative, and the result is trivial with $\beta = 0$. Consider then the case $L < 0$.

Since~$\Scal$ is a smooth compact boundary, the Dirichlet trace operator ${\gamma_D}: \Hcal^1(\Scal') \to \LL^2(\Scal)$ is bounded. Moreover, for any $\eps > 0$, the following trace inequality holds (often derived locally by integrating along the normal vector in a tubular neighborhood of the boundary; see, for example, Theorem~1.5.1.10 in~\cite{grisvardBook}, which applies here because the boundary~$\Scal$ is compact)
\begin{equation} \label{eq:trace_ineq}
 \int_{\Scal} |\gamma_D(u)|^2 {\dr \sigma}\leq \eps \int_{\Scal'} |\nabla |u||^2 {\dr \bx} + C_\eps \int_{\Scal'} |u|^2 {\dr \bx},
\end{equation}
where $C_\eps$ is a constant depending on $\eps$ and the geometry of~$\Scal$.

We invoke the diamagnetic inequality~\cite[Section 7.21]{LiebLoss1997}, which holds for more general $\bA$ with components in~$\LL^2_{\mathrm{loc}}(\mathcal S')$ and for $u\in\LL^2(\mathcal S')$ with $(\nabla + i \bA)u\in \LL^2(\mathcal S')$, and it states that for almost every~$\bx$,
\[
|\nabla |u|(\bx)| \leq |(\nabla + i\bA)u(\bx)|.
\]
Substituting this into \eqref{eq:trace_ineq}, we obtain
\[
\int_{\Scal} |\gamma_D(u)|^2 {\dr \sigma}\leq \eps \norm{(\nabla + i\bA)u}_{\LL^2(\Scal')}^2 + C_\eps \norm{u}_{\LL^2(\Scal')}^2.
\]
Now, estimating the form $b^{(L)}(u)$, for $L < 0$,
\begin{align*}
 b^{(L)}(u) &= \norm{(\nabla+i\bA)u}^2 - |L| \int_{\Scal} |\gamma_D(u)|^2 {\dr \sigma}\\
 &\geq \norm{(\nabla+i\bA)u}^2 - |L| \left( \eps \norm{(\nabla+i\bA)u}^2 + C_\eps \norm{u}^2 \right) \\
 &= (1 - |L|\eps)\, \norm{(\nabla+i\bA)u}^2 - |L| C_\eps \norm{u}^2.
\end{align*}
We choose $\eps$ sufficiently small such that $|L|\eps < 1$ so that the coefficient of the kinetic term is positive. Dropping the positive kinetic term, one gets
\[
b^{(L)}(u) \geq - |L| C_\eps \norm{u}_{\LL^2(\Scal')}^2.
\]
Thus, the form is bounded from below with $\beta = -|L|C_\eps$.

\subsection*{Closedness}
A lower bounded form $b^{(L)}$ is closed if the domain $D(b^{(L)}) = \Hcal^1(\Scal')$ is complete with respect to the form norm defined by~\cite{ISTQD}
\[
\norm{u}_{b^{(L)}} := \left( b^{(L)}(u) + (1 - \beta) \norm{u}_{\LL^2(\Scal')}^2 \right)^{1/2}.
\]
Since $\Hcal^1(\Scal')$ is complete under the standard Sobolev norm $\norm{\cdot}_{\Hcal^1}$, it suffices to show that $\norm{u}_{b^{(L)}}$ is equivalent to $\norm{u}_{\Hcal^1(\Scal')_A}$, where $\norm{u}_{\Hcal^1_A}^2 = \norm{(\nabla+i\bA)u}^2 + \norm{u}^2$. Note that if $\bA$ is bounded, this is equivalent to the standard $\Hcal^1$ norm.

{\em Upper bound.} Using the continuity of the trace map, there exists $C_{\mathrm{tr}} > 0$ such that $\int_{\Scal} |\gamma_D(u)|^2 {\dr \sigma}\leq C_{\mathrm{tr}} \norm{u}_{\Hcal^1_A}^2$. Thus
\[
b^{(L)}(u) \leq \norm{(\nabla+i\bA)u}^2 + |L| C_{\mathrm{tr}} \norm{u}_{\Hcal^1_A}^2 \leq (1 + |L|C_{\mathrm{tr}}) \norm{u}_{\Hcal^1_A}^2.
\]
 Adding the $(1-\beta) \norm{u}^2$ term implies $\norm{u}_{b^{(L)}} \lesssim \norm{u}_{\Hcal^1_A}$.

{\em Lower bound.} Using the estimate from the lower boundedness proof with the specific choice of $\eps$ such that $(1 - |L|\eps) > 0$, one gets
\begin{align*}
 b^{(L)}(u) + (1-\beta) \norm{u}^2 &\geq (1-|L|\eps)\norm{(\nabla+i\bA)u}^2 - |L|C_\eps \norm{u}^2 \\ &+ (1 + |L|C_\eps) \norm{u}^2 \\
 &\geq c \left( \norm{(\nabla+i\bA)u}^2 + \norm{u}^2 \right) \\
 &= c \norm{u}_{\Hcal^1_A}^2.
\end{align*}
Since the norms are equivalent and $\Hcal^1(\Scal')$ is a Banach space, the form is closed.

 \subsection*{Characterization of $H^{(L)}$} 
 Now we pass to the characterization of the associated self-adjoint operator. Let $u \in \dom H^{(L)} \subseteq \Hcal^1(\Scal')$ and let $v \in \Hcal^1(\Scal')$. The operator $H^{(L)}$ satisfies the relation
 \begin{equation} \label{eq:relation}
 b^{(L)}(u, v) = \langle H^{(L)}u, v \rangle_{\LL^2(\Scal')},
 \end{equation}
with
 \begin{equation} \label{eq:sesquilinear}
 b^{(L)}(u, v) = \int_{\Scal'}\overline{ (\nabla + i\bA)u} \cdot {(\nabla + i\bA)v}\,\dr \bx + L \int_{\Scal} \overline{\gamma_D(u)} {\gamma_D(v)}\,\dr \sigma.
 \end{equation}
 Let $\nabla_{\bA} = \nabla + i\bA$. Introduce the formal differential operator acting in the distribution sense on the domain~$\Scal'$ as $H_{\text{for}} = -\nabla_{\bA}^2$. We apply the magnetic Green's first identity (integration by parts) to the first term of \eqref{eq:sesquilinear}.
 
 Recall that the Divergence Theorem relates the integral over a domain to the integral over its boundary~$\Scal$ using the {outward} pointing normal, which here is the vector~$-\nu$.

 Green's identity gives
 \[
 \int_{\Scal'} \overline{\nabla_{\bA} u} \cdot {\nabla_{\bA} v}\,\dr \bx 
 = \int_{\Scal'} -\overline{(\nabla_{\bA}^2 u)} {v}\,\dr \bx + \int_{\Scal} \overline{(-\nu \cdot \nabla_{\bA} u)} {v}\,\dr \sigma.
\]
 Inserting this into the expression for $b^{(L)}(u, v)$, one finds
 \begin{align*}
 b^{(L)}(u, v) &= \int_{\Scal'} -\overline{(\nabla_{\bA}^2 u)} {v}\,\dr \bx - \int_{\Scal} \overline{(\nu \cdot \nabla_{\bA} u)} {\gamma_D(v)}\,\dr \sigma + L \int_{\Scal} \overline{\gamma_D(u)} {\gamma_D(v)}\,\dr \sigma \\
 &= \int_{\Scal'} \overline{H_{\text{for}} u} \, {v}\,\dr \bx + \int_{\Scal} \left(\overline{ L \gamma_D(u) - \nu \cdot \nabla_{\bA} u} \right) {\gamma_D(v)}\,\dr \sigma.
 \end{align*}
 For \eqref{eq:relation} to hold for all $v \in \Hcal^1(\Scal')$, the term $\langle H^{(L)}u, v \rangle_{\LL^2}$ must correspond to the volume integral $\int_{\Scal'} H_{\text{for}} u \, \overline{v}\,\dr \bx$. This implies two conditions:
 
 \begin{enumerate}
 \item The action of the operator is $H^{(L)} u = -(\nabla + i\bA)^2 u$ in~$\Scal'$.
 \item The boundary integral must vanish for all ``test functions'' $v$, that is,
 \[
 \int_{\Scal} \left( \overline{L \gamma_D(u) - \nu \cdot \nabla_{\bA} u} \right) {\gamma_D(v)}\,\dr \sigma = 0.
 \]
 \end{enumerate}
 Since the trace of $v$ on~$\Scal$ is arbitrary (the trace map from $\Hcal^1(\Scal')$ to $\Hcal^{1/2}(\Scal)$ is surjective), the term in the parenthesis must be zero almost everywhere on~$\Scal$. Thus, we obtain the magnetic Robin boundary condition:
 \[
 \nu \cdot (\nabla + i\bA)u = L \gamma_D(u) \quad \text{on } \Scal.
 \]

\section{The quadratic form of $h^{(L)}$}\label{appActOphL}
Let $u \in \dom h^{(L)}$ and $v \in \Hcal^1(\R^3)$ be a test function. Then
 \[
 q^{(L)}(u, v) = \langle h^{(L)}u, v \rangle_{\LL^2(\R^3)}.
\] The sesquilinear form is
 \[
 q^{(L)}(u, v) = \int_{\R^3} \overline{(\nabla + i\mathbf{A})u} \cdot {(\nabla + i\mathbf{A})v}\, \dr\bx + L \int_{\Scal} \overline{\gamma_D(u)} {\gamma_D(v)}\, \dr\sigma.
 \]
 We split the volume integral into the interior~$\Scal^\circ$ and the exterior~$\Scal'$. Let $\nabla_A = \nabla + i\mathbf{A}$. We apply Green's first identity on each subdomain.
 
 {\em Interior domain ($\Scal^\circ$).}
 The unit normal pointing {out} of~$\Scal^\circ$ is exactly $\nu$, so
 \[
 \int_{\Scal^\circ} \overline{\nabla_{\bA} u} \cdot {\nabla_{\bA} v}\, \dr\bx = \int_{\Scal^\circ} -\overline{(\nabla_{\bA}^2 u)} \, v\, \dr\bx + \int_{\Scal} \overline{(\nu \cdot \nabla_{\bA} u)}_- {\gamma_D(v)}\, \dr\sigma.
 \]
 Here, $(\cdot)_-$ denotes the trace taken from the interior.
 
 {\em Exterior domain ($\Scal'$).}
 The unit normal pointing {out} of~$\Scal'$ is $-\nu$, so
 \[
 \int_{\Scal'} \overline{\nabla_{\bA} u} \cdot {\nabla_{\bA} v}\, \dr\bx = \int_{\Scal'} -\overline{(\nabla_{\bA}^2 u)}\, {v}\, \dr\bx + \int_{\Scal} \overline{(-\nu \cdot \nabla_{\bA} u)}_+ {\gamma_D(v)}\, \dr\sigma.
 \]
 Here, $(\cdot)_+$ denotes the trace taken from the exterior.

 Adding the two contributions
 \begin{align*}
 \int_{\R^3} \overline{\nabla_{\bA} u} \cdot {\nabla_{\bA} v}\, \dr\bx &= \int_{\R^3 \setminus \Scal} -\overline{(\nabla_{\bA}^2 u)}\, {v}\, \dr\bx + \int_{\Scal} \overline{\left( (\nu \cdot \nabla_{\bA} u)_- - (\nu \cdot \nabla_{\bA} u)_+ \right)} \,{\gamma_D(v)}\, \dr\sigma \\
 &= \int_{\R^3 \setminus \Scal} -\overline{(\nabla_{\bA}^2 u)}\, {v}\, \dr\bx - \int_{\Scal}\left[ \overline{\nu \cdot(\nabla+i \mathbf{A}) u} \right]_{\Scal} {\gamma_D(v)}\, \dr\sigma.
 \end{align*}
 Substituting this into the expression for $q^{(L)}(u, v)$, one finds
 \begin{align*}
 q^{(L)}(u, v) &= \int_{\R^3 \setminus \Scal} -\overline{(\nabla_{\bA}^2 u)} \,{v}\, \dr\bx - \int_{\Scal}\left[ \overline{\nu \cdot(\nabla+i \mathbf{A}) u} \right]_{\Scal} \,{\gamma_D(v)}\, \dr\sigma + L \int_{\Scal} \overline{\gamma_D(u)} \,{\gamma_D(v)}\, \dr\sigma \\
 &= \langle -\nabla_{\bA}^2 u, v \rangle_{\LL^2(\R^3 \setminus \Scal)} + \int_{\Scal} \left(\overline{ L \gamma_D(u) -\left[ \nu \cdot(\nabla+i \mathbf{A}) u \right]_{\Scal}} \right) {\gamma_D(v)}\, \dr\sigma.
 \end{align*}
 For this to equal $\langle h^{(L)}u, v \rangle_{\LL^2(\R^3)}$ for all $v \in \Hcal^1(\R^3)$, the boundary term must vanish. Since the trace map $\gamma_D:\Hcal^1(\R^3) \to \Hcal^{1/2}(\Scal)$ is surjective, $\gamma_D(v)$ is arbitrary. Therefore,
 \[
 \left[ \nu \cdot(\nabla+i \mathbf{A}) u \right]_{\Scal} = L \gamma_D(u).
 \]

\

\noindent{\bf Acknowledgments.} The author  acknowledges partial financial support from CNPq (Brazilian National Council for Scientific and Technological Development) under Grant No.~303689/2021-8.

\end{document}